%% file: BGLR14.tex
\documentclass[noinfoline,stslayout,english]{imsart} 
\usepackage{amsfonts,amssymb,amsmath,amscd,latexsym}
\usepackage{float}
\usepackage{amsthm}
\usepackage{esint}
\usepackage{graphicx}
\usepackage[round,sort,comma,authoryear]{natbib}
\usepackage[]{units}
\usepackage{url}
\usepackage{algpseudocode}
\usepackage{qtree}
\usepackage{algorithm}
\usepackage{comment}
\usepackage{color}

\newtheorem{exemp}{Example}[section]

\newtheorem{remark}{Remark}

  \theoremstyle{plain}
  \newtheorem{lem}{\protect\lemmaname}
  \theoremstyle{remark}
  
  \theoremstyle{definition}
  \newtheorem{condition}{\protect\conditionname}
  \theoremstyle{plain}
  \newtheorem{prop}{\protect\propositionname}

\newcommand\dd{\mathrm{d}}
\newcommand\varm{\text{var}}
\newcommand\bigo{\text{O}}

\DeclareMathOperator*{\argmax}{arg\,max}

\usepackage{babel}
  \providecommand{\conditionname}{Condition}
  \providecommand{\lemmaname}{Lemma}
  \providecommand{\propositionname}{Proposition}
  \providecommand{\remarkname}{Remark}

\begin{document}

\begin{frontmatter}

\title{Accelerating Metropolis--Hastings algorithms by Delayed Acceptance\protect\thanksref{T1}}
\runtitle{Delayed Acceptance for Metropolis--Hastings}
\thankstext{T1}{This document has been typeset using the IMS {\sf imsart} \LaTeX~style file, but it has not been
previously submitted to any IMS journal.}

\begin{aug}
 \author{\snm{{\sc Marco Banterle}}}
 \affiliation{Universit{\'e} Paris-Dauphine, CEREMADE, and CREST, Paris}
 \author{\snm{{\sc Clara Grazian}}}
 \affiliation{Sapienza Universit\`a di Roma, Universit{\'e} Paris-Dauphine, CEREMADE, and CREST, Paris}
 \author{\snm{{\sc Anthony Lee}}}
 \affiliation{University of Warwick, Coventry}
 \author{\snm{{\sc Christian P.~Robert}}}
 \affiliation{Universit{\'e} Paris-Dauphine, CEREMADE, Dept. of Statistics, University of Warwick, and CREST, Paris}
\end{aug}

%
%
%
%

\begin{abstract} MCMC algorithms such as Metropolis--Hastings algorithms are
slowed down by the computation of complex target distributions as exemplified
by huge datasets. We offer in this paper a useful generalisation of the Delayed Acceptance approach, 
devised to reduce the computational costs of such algorithms by a simple and universal
divide-and-conquer strategy. The idea behind the generic acceleration is to
divide the acceptance step into several parts, aiming at a major reduction in
computing time that out-ranks the corresponding reduction in acceptance
probability.  
Each of the components can be sequentially compared with a uniform
variate, the first rejection signalling that the proposed value is considered no
further. We develop moreover theoretical bounds for the variance of associated 
estimators with respect to the variance of the standard Metropolis--Hastings and 
detail some results on optimal scaling and general optimisation of the procedure.
We illustrate those accelerating features on a series of examples.  \end{abstract}

\end{frontmatter}
\noindent {\bf Keywords:} Large Scale Learning and Big Data, MCMC,
likelihood function, acceptance probability,
mixtures of distributions, Jeffreys prior

\section{Introduction}\label{intro} \label{sec:zero_zero}

When running an MCMC sampler such as Metropolis--Hastings algorithms
\citep{robert:casella:2004}, the complexity of the target density required by
the acceptance ratio may lead to severe slow-downs in the execution of the
algorithm. A direct illustration of this difficulty is the simulation from a
posterior distribution involving a large dataset of $n$ points for which the
computing time is at least of order $\bigo(n)$. Several solutions to this issue have
been proposed in the recent literature \citep{korattikara:chen:welling:2013,
neiswanger:wang:xing:2013, scott:etal:2013,wang:dunson:2013}, taking advantage
of the likelihood decomposition 
\begin{equation}\label{eq:baselike}
\prod_{i=1}^n \ell(\theta|x_i) 
\end{equation}
to handle subsets of the data on different processors (CPU), graphical units (GPU), or
even computers.  However, there is no consensus on the method of choice, some
leading to instabilities by removing most prior inputs and others to
approximations delicate to evaluate or even to implement.

Our approach here is to delay acceptance (rather than rejection as in
\cite{tierney:mira:1998}) by sequentially comparing parts of the MCMC acceptance
ratio to independent uniforms, in order to stop earlier the computation of the
aforesaid ratio, namely as soon as one term is below the corresponding uniform.

More formally, consider a generic Metropolis--Hastings algorithm where the acceptance ratio $\pi(y)q(y,x) / \pi(x)q(x,y)$ is 
compared with a $\,\mathcal{U}(0,1)$ variate to decide whether or not the Markov chain switches from the current value $x$ to the proposed value $y$
\citep{robert:casella:2004}. If we now decompose the ratio as an arbitrary product
\begin{equation}\label{eq:prodeck}
{\pi(y)\,q(y,x)}\big/{\pi(x)q(x,y)} = \prod_{k=1}^d \rho_k(x,y)\,
\end{equation}
where the only constraint is that the functions $\rho_k$ are all positive and satisfy the balance condition $\rho_k(x,y)
= \rho_k(y,x)^{-1}$ and then accept the move with probability 
\begin{equation}\label{eq:prodike}
\prod_{k=1}^d \min\left\{\rho_k(x,y),1\right\}\,,
\end{equation}
i.e. by successively comparing uniform variates $u_k$ to the terms $\rho_k(x,y)$,
the motivation for our delayed approach is that the same target density $\pi(\cdot)$ is stationary for the resulting Markov chain. 

The mathematical validation of this simple if surprising result can be seen as a consequence of
\cite{christen:fox:2005}. This paper reexamines \cite{fox:nicholls:1997}, where the idea of testing for acceptance using
an approximation and before computing the exact likelihood was first suggested.  In \cite{christen:fox:2005}, the
original proposal density $q$ is used to generate a value $y$ that is tested against an approximate target $\tilde \pi$.
If accepted, $y$ can be seen as coming from a pseudo-proposal $\tilde q$ that simply is formalising the earlier preliminary step and is then tested against the true target $\pi$. The validation in \cite{christen:fox:2005} follows from standard detailed balance arguments; we will focus formally on this point in Section \ref{sec:convergence}. 

In practice, sequentially comparing those probabilities with $d$ uniform variates means
that the comparisons stop at the first rejection, implying a gain in computing time
if the most costly items in the product \eqref{eq:prodeck} are saved for the final comparisons.

Examples of the specific two-stage Delayed Acceptance as defined by \cite{christen:fox:2005} can be found in
\cite{golightly:2014}, in the pMCMC context, and in \cite{shestopaloff:neal:2013}.

The major drawback of the scheme is that Delayed Acceptance \emph{efficiently}
reduces the computing cost only when the approximation $\tilde \pi$ is ``good enough"  or ``flat enough", since the
probability of acceptance of a proposed value will always be smaller than in the original Metropolis--Hastings scheme.
In other words, the original Metropolis--Hastings kernel dominates the new one in Peskun's \citep{peskun:1973} sense. 
The most relevant question raised by \cite{christen:fox:2005} is thus how to achieve a proper approximation; note in
fact that while in Bayesian statistics a decomposition of the target is always available, by breaking original data in
subsamples and considering the corresponding likelihood parts or even by just separating the prior, proposal and
likelihood ratio into different factors, these decompositions may just lead to a deterioration of the algorithm
properties without impacting the computational efficiency.


\begin{figure}
\begin{center}
\includegraphics[width=.4\textwidth]{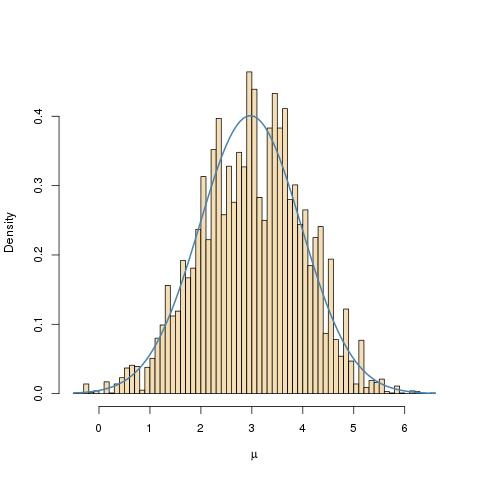}
\end{center}
\caption{\label{fig:norlayed}
Fit of a two-step Metropolis--Hastings algorithm applied to a normal-normal
posterior distribution $\mu|x\sim N(x/(\{1+\sigma_\mu^{-2}\},1/\{1+\sigma_\mu^{-2}\})$
when $x=3$ and $\sigma_\mu=10$, based on $T=10^5$ iterations and a first acceptance
step considering the likelihood ratio and a second acceptance step considering the prior ratio,
resulting in an overall acceptance rate of 12\%
}\end{figure}


\begin{figure}
\includegraphics[width=.3\textwidth]{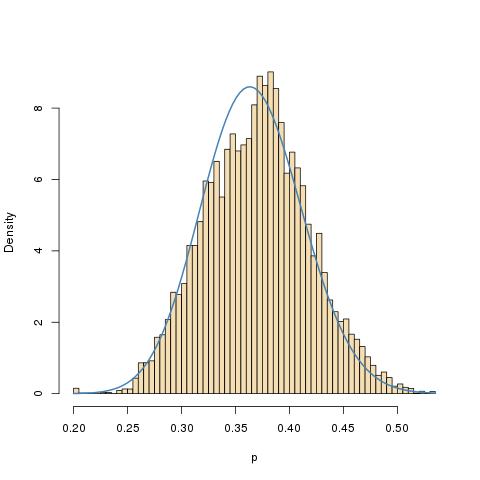}
\includegraphics[width=.3\textwidth]{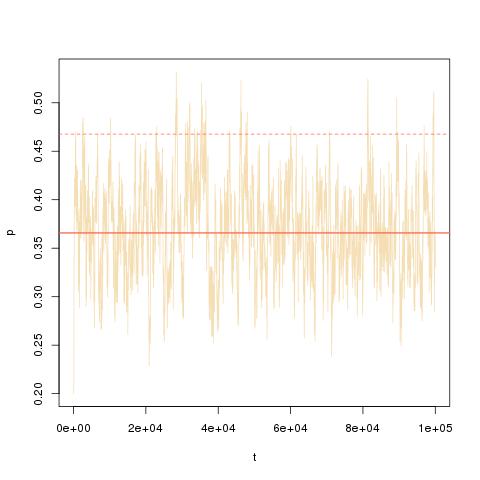}
\includegraphics[width=.3\textwidth]{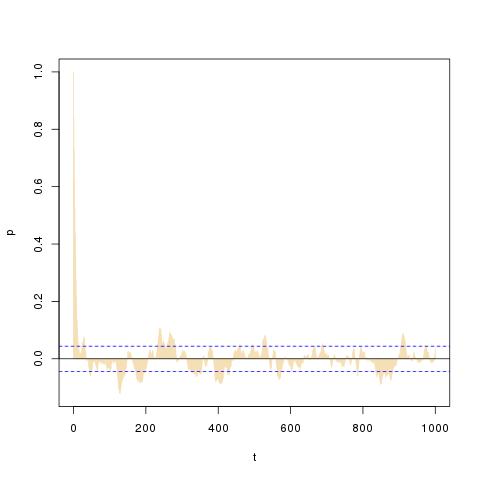}
\caption{\label{fig:binomial}
(left)
Fit of a multiple-step Metropolis--Hastings algorithm applied to a Beta-binomial
posterior distribution $p|x\sim Be(x+a,n+b-x)$ when $N=100$, $x=32$, $a=7.5$    
and $b=.5$. The binomial $\mathcal{B}(N,p)$ likelihood is replaced with a product of $100$ Bernoulli terms and
an acceptance step is considered for the ratio of each term. The histogram is based on $10^5$
iterations, with an overall acceptance rate of 9\%; 
(centre) raw sequence of successive values of $p$ in the Markov chain simulated 
in the above experiment; (right) autocorrelogram of the above sequence.
}\end{figure}

However, even in these simple cases, it is possible to find examples where Delayed Acceptance may be profitable.
Consider for instance resorting to a costly non-informative prior distribution (as illustrated in Section \ref{sec:mixis} in the
case of mixtures); 
here the first acceptance step can be solely based on the ratio of the likelihoods and the second step, which involves the ratio of the priors, does not require to be computed when the first test leads to rejection. 
Even more often, the converse decomposition applies to complex or
just costly likelihood functions, in that the prior ratio may first be used to
eliminate values of the parameter that are too unlikely for the prior density.
As shown in Figure \ref{fig:norlayed}, a standard normal-normal example confirms that the true posterior and the
histogram resulting from such a simulated sample are in agreement.

In more complex settings, as for example in ``Big Data" settings where the likelihood is made of a very large number of
terms, the above principle also applies to any factorisation of the like of \eqref{eq:baselike} so that each individual
likelihood factor can be evaluated separately. This approach increases both the variability of the
evaluation and the potential for rejection, but, if each term of the factored likelihood is
sufficiently costly to compute, the decomposition brings some improvement in
execution time. The graphs in Figure \ref{fig:binomial} illustrate
an implementation of this perspective in the Beta-binomial case, namely when
the binomial $\mathcal{B}(N,p)$ observation $x=32$ is replaced with a sequence of $N$ Bernoulli
observations. The fit is adequate on $10^5$ iterations, but the autocorrelation
in the sequence is very high (note that the ACF is for the 100 times thinned
sequence) while the acceptance rate falls down to 9\%. (When the original
$y=32$ observation is (artificially) divided into 10, 20, 50, and 100 parts,
the acceptance rates are 0.29, 0.25, 0.12, and 0.09, respectively.) The gain in using this decomposition is only
appearing when each Bernoulli likelihood computation becomes expensive enough.

On one hand, the order in which the product \eqref{eq:prodike} is explored determines the computational efficiency of the scheme, while, on the other hand, it has no impact on the overall convergence of the resulting Markov chain, since the acceptance of a proposal does require computing all likelihood values. It therefore makes sense to try to optimise this order by ranking the entries in a way that improves the execution speed of the algorithm (see Section \ref{sec:optim_rank}).

We also stress that the Delayed Acceptance principle remains valid even when the likelihood function or the prior are not integrable
over the parameter space. Therefore, the prior may well be improper. For instance, when the prior distribution
is constant, a two-stage acceptance scheme reverts to the original Metropolis--Hastings one.

Finally, while the Delayed Acceptance methodology is intended to cater to complex likelihoods or priors, it does not bring a solution {\em per se} to the ``Big Data'' problem in that (a) all terms in the product must eventually be computed; and (b) all terms previously computed (i.e., those computed for the last accepted value of the parameter) 
must be either stored for future comparison or recomputed. See, e.g., \cite{scott:etal:2013,wang:dunson:2013}, for
recent entries on different parallel ways of handling massive datasets.

The plan of the paper is as follows: in Section \ref{sec:convergence}, we validate the decomposition of the acceptance
step into a sequence of decisions, arguing about the computational gains brought by this generic modification of
Metropolis--Hastings algorithms and further analysing the relation between the proposed method and the
Metropolis--Hastings algorithm in terms of convergence properties and asymptotic variances of statistical estimates.  In Section \ref{sec:others} we briefly state the relations between Delayed Acceptance and other methods present in the literature.  In Section \ref{sec:optim} we aim at giving some intuitions on how to improve the behaviour of Delayed Acceptance by ranking the factors in a given decomposition to achieve optimal computational efficiency and finally give some preliminary results in terms of optimal scaling for the proposed method.  Then Section \ref{sec:examples} studies Delayed Acceptance within three realistic environments, the first one made of logistic regression targets, the second one alleviating the computational burden from a Geometric Metropolis adjusted Langevin algorithm and a third one handling an original analysis of a parametric mixture model via genuine Jeffreys priors. Section \ref{sec:conclusion} concludes the paper.

\section{Validation and convergence of Delayed Acceptance} \label{sec:convergence}

In this section, we establish that Delayed Acceptance is a valid Markov chain Monte Carlo scheme and analyse on a
theoretical basis the differences with the original version.

\subsection{The general scheme}

We assume for simplicity that the target distribution $\pi$ and the proposal distributions $Q(x,\cdot)$ all admit
densities w.r.t. the Lebesgue or counting measures. We also denote  by $\pi$ the target density and let $q(x,y)$ denote
the proposal density.

Let $(X_{n})_{n\geq1}$ be a Markov chain evolving on $\mathsf{X}$
with Metropolis--Hastings Markov transition kernel $P$ associated with $q$ and $\pi$,
i.e. for $A \in \mathcal B(\mathsf{X})$
\[
P(x,A):=\int_{A}q(x,y)\alpha(x,y){\rm d}y+\left(1-\int_{\mathsf{X}}q(x,y)\alpha(x,y){\rm d}y\right)\mathbf{1}_{A}(x),
\]
where 
\[
\alpha(x,y):=1\wedge r(x,y),\qquad r(x,y):=\frac{\pi(y)q(y,x)}{\pi(x)q(x,y)}.
\]
Above, $\alpha(x,y)$ is known as the Metropolis--Hastings acceptance probability and $r(x,y)$ as the
Metropolis--Hastings acceptance ratio.

We consider the class of ``Delayed acceptance'' Markov kernels associated
with $P$, which are defined by factorisations of the function $r$ as
\begin{equation}
r(x,y)=\prod_{k=1}^{d}\rho_{k}(x,y)\label{eq:factorize_r}
\end{equation}
with all components in the product satisfying $\rho_{k}(x,y)=\rho_{k}(y,x)^{-1}$. The associated Delayed Acceptance
Markov kernel is then defined as
\[
\tilde{P}(x,A):=\int_{A}q(x,y)\tilde{\alpha}(x,y){\rm d}y+\left(1-\int_{\mathsf{X}}q(x,y)\tilde{\alpha}(x,y){\rm d}y\right)\mathbf{1}_{A}(x),
\]
where
\[
\tilde{\alpha}(x,y):=\prod_{k=1}^{d}\{1\wedge\rho_{k}(x,y)\}.
\]
We will denote by $(\tilde{X}_{n})_{n\geq1}$ the Markov chain associated with $\tilde{P}$. 

The order in which the sequence of functions $\rho_{k}$ appears in the factorisation \eqref{eq:factorize_r} is important for algorithmic specification, as
can be seen in Algorithm~\ref{alg:Delayed-acceptance}.  It means that $\rho_{1}(x,Y)$ is evaluated first,
$\rho_{2}(x,Y)$ second, and so on until $\rho_{d}(x,Y)=r(x,Y)/\prod_{k=1}^{d-1}\rho_{k}(x,Y)$ which is last, 
with the motivation that ``early rejection'' can allow computational savings by avoiding the computation of the
subsequent $\rho_{k}(x,Y)$.

\begin{algorithm}[H]

\protect\caption{Delayed Acceptance\label{alg:Delayed-acceptance}}
To sample from $\tilde{P}(x,\cdot)$:
\begin{enumerate}
\item Sample $Y\sim Q(x,\cdot)$.
\item For $k=1,\ldots,d$:

\begin{itemize}
\item With probability $1\wedge\rho_{k}(x,Y)$ continue, otherwise stop
and output $x$.
\end{itemize}
\item Output $Y$.\end{enumerate}
\end{algorithm}

\subsection{Validation}

The first lemma is a standard representation leading to the validation of the Delayed Acceptance Markov chain:
\begin{lem}\label{lem:reversible_ratio_accept_probs}
For any Markov chain with transition kernel $\Pi$ of the form
\[
\Pi(x,A)=\int_{A}q(x,y)a(x,y){\rm d}y+\left(1-\int_{\mathsf{X}}q(x,y)a(x,y){\rm d}y\right)\mathbf{1}_{A}(x),
\]
and satisfying detailed balance, the function $a(\cdot)$ satisfies (for
$\pi$-a.a. $x,y$)
\[
\frac{a(x,y)}{a(y,x)}=r(x,y).
\]
\end{lem}
\begin{proof}
This follows immediately from the detailed balance condition
\[
\pi(x)q(x,y)a(x,y)=\pi(y)q(y,x)a(y,x).
\]
\end{proof}
The Delayed Acceptance Markov chain $(\tilde{X}_{n})_{n\geq1}$ is then associated with the intended target:
\begin{lem}
$(\tilde{X}_{n})_{n\geq1}$ is a $\pi$-reversible Markov chain.\end{lem}
\begin{proof}
From Lemma~\ref{lem:reversible_ratio_accept_probs} it suffices to
verify that $\tilde{\alpha}(x,y)/\tilde{\alpha}(y,x)=r(x,y)$. Indeed,
we have
\begin{eqnarray*}
\frac{\tilde{\alpha}(x,y)}{\tilde{\alpha}(y,x)} & = &
\frac{\prod_{k=1}^{d}\{1\wedge\rho_{k}(x,y)\}}{\prod_{k=1}^{d}\{1\wedge\rho_{k}(y,x)\}}\\
 & = & \prod_{k=1}^{d}\frac{1\wedge\rho_{k}(x,y)}{1\wedge\rho_{k}(y,x)}\\
 & = & \prod_{k=1}^{d}\rho_{k}(x,y)=r(x,y),
\end{eqnarray*}
since $\rho_{k}(y,x)=\rho_{k}(x,y)^{-1}$ and $(1\wedge a)/(1\wedge a^{-1})=a$ for $a\in\mathbb{R}_{+}$.\end{proof}
\begin{remark}
It is immediate to show that
$$\tilde{\alpha}(x,y)=\prod_{k=1}^{d}\{1\wedge\rho_{k}(x,y)\}\leq1\wedge\prod_{k=1}^{d}\rho_{k}(x,y)=1\wedge r(x,y)=\alpha(x,y),$$
since $(1\wedge a)(1\wedge b)\leq(1\wedge ab)$ for $a,b\in\mathbb{R}_{+}$.\end{remark}


\subsection{Comparisons of the kernels $P$ and $\tilde{P}$}

Given a probability measure $\mu$, let us denote
\begin{align*}
\mu(f)&:=\int_{\mathsf{E}}f(x)\mu({\rm d}x)\,,\quad L^{2}(\mathsf{E},\mu):=\{f:\mu(f^{2})<\infty\}\\
L_{0}^{2}\bigl(\mathsf{E},\mu\bigr)&:=\bigl\{ f\in L^{2}(\mathsf{E},\mu):\mu\bigl(f\bigr)=0\bigr\}\,.
\end{align*}
For a generic Markov kernel $\Pi:\mathsf{E}\times\mathcal{B}(\mathsf{E})$
with unique invariant probability measure $\mu$, we define
\[
{\rm var}(f,\Pi):=\lim_{n\rightarrow\infty}{\rm var}\left(n^{-\frac{1}{2}}\sum_{i=1}^{n}\left[f(X_{i})-\mu(f)\right]\right),
\]
where $(X_{n})_{n\geq1}$ is a Markov chain with Markov kernel $\Pi$
initialised with $X_{1}\sim\mu$.
\begin{remark}
One can immediately conclude from the construction of $\tilde{P}$
that ${\rm var}(f,P)\leq{\rm var}(f,\tilde{P})$ for any $f\in L^{2}(\mathsf{X},\pi)$,
using Peskun ordering \citep{peskun73,tierney98}, since $\tilde{\alpha}(x,y)\leq\alpha(x,y)$
for any $(x,y)\in\mathsf{X}^{2}$.
\end{remark}
For any $f\in L^{2}(\mathsf{E},\mu)$ we define the Dirichlet form
associated with a $\mu$-reversible Markov kernel $\Pi:\mathsf{E}\times\mathcal{B}(\mathsf{E})$
as
\[
\mathcal{E}_{\Pi}(f):=\frac{1}{2}\int\mu(\mathrm{d}x)\Pi(x,\mathrm{d}y)\left[f(x)-f(y)\right]^{2}.
\]
The (right) spectral gap of a generic $\mu$-reversible Markov kernel
has the following variational representation 
\[
{\rm Gap}\left(\Pi\right):=\inf_{f\in L_{0}^{2}(\mathsf{E},\mu)}\frac{\mathcal{E}_{\Pi}(f)}{\left\langle f,f\right\rangle _{\mu}}\quad.
\]
which leads to the following domination lemma:
\begin{lem}[{\citep[Lemma~34]{andrieuLee:13}}]
\label{lem:alv_dirichet_comparison}Let $\Pi_{1}$ and $\Pi_{2}$
be $\mu$-reversible Markov transition kernels of $\mu$-irreducible
and aperiodic Markov chains, and assume that there exists $\varrho>0$
such that for any $f\in L_{0}^{2}\bigl(\mathsf{E},\mu\bigr)$ 
\[
\mathcal{E}_{\Pi_{2}}\bigl(f\bigr)\geq\varrho\mathcal{E}_{\Pi_{1}}\bigl(f\bigr)\quad,
\]
then
\[
{\rm Gap}\bigl(\Pi_{2}\bigr)\geq\varrho{\rm Gap}\bigl(\Pi_{1}\bigr)
\]
and
\[
{\rm var}\bigl(f,\Pi_{2}\bigr)\leq(\varrho^{-1}-1){\rm var}_{\mu}(f)+\varrho^{-1}{\rm var}\left(f,\Pi_{1}\right)\quad f\in L_{0}^{2}(\mathsf{E},\mu).
\]

\end{lem}

We will need the following condition in the sequel:
\begin{condition}
\label{cond:rho_klower}Defining $A:=\{(x,y)\in\mathsf{X}^{2}:r(x,y)\geq1\}$, there exists a $c$ such that\\
\centerline{$\inf_{(x,y)\in A}\min_{k\in\{1,\ldots,d\}}\rho_{k}(x,y)\geq c$.}\end{condition}
\begin{prop}
\label{prop:generic_da_result}Assume Condition~\ref{cond:rho_klower}.
Then Lemma~\ref{lem:alv_dirichet_comparison} holds with $\Pi_{1}=P$,
$\Pi_{2}=\tilde{P}$, $\mu=\pi$ and $\varrho=c^{d-1}$.\end{prop}
\begin{proof}
Since $r(x,y)\geq1$, we have $\alpha(x,y)=1$. On the other hand, for $(x,y)\in A$
\begin{eqnarray*}
\tilde{\alpha}(x,y) & = & \prod_{k=1}^{d}1\wedge\rho_{k}(x,y)\\
 & = & \prod_{k:\rho_{k}(x,y)<1}\rho_{k}(x,y)
 \geq c^{|\{k:\rho_{k}(x,y)<1\}|}\geq c^{d-1},
\end{eqnarray*}
since at least one $\rho_{k}(x,y)\geq1$ whenever $r(x,y)\geq1$.

From Lemma~\ref{lem:reversible_ratio_accept_probs}, when $(x,y)\in A$, we have 
$$\tilde{\alpha}(y,x)=\tilde{\alpha}(x,y)/r(x,y)\geq c^{d-1}\alpha(x,y)/r(x,y)=c^{d-1}\alpha(y,x)$$
and thus $\tilde{\alpha}(x,y)\geq c^{d-1}\alpha(x,y)$ for any $(x,y)\in\mathsf{X}^{2}$.
It follows that
\begin{eqnarray*}
\mathcal{E}_{\tilde{P}}(f) & = & \int_{\mathsf{X}}\pi({\rm d}x)\tilde{P}(x,{\rm d}y)\left(f(x)-f(y)\right)^{2}\\
 & = & \int_{\mathsf{X}}\pi({\rm d}x)P(x,{\rm d}y)\frac{\tilde{\alpha}(x,y)}{\alpha(x,y)}\left(f(x)-f(y)\right)^{2}\\
 & \geq & c^{d-1}\int_{\mathsf{X}}\pi({\rm d}x)P(x,{\rm d}y)\left(f(x)-f(y)\right)^{2}
   = c^{d-1}\mathcal{E}_{P}(f),
\end{eqnarray*}
and we conclude.
\end{proof}
The implication of this result is that, if $P$ admits a right spectral
gap, then so does $\tilde{P}$, whenever Condition~\ref{cond:rho_klower}
holds. Furthermore, and irrespective of whether or not $P$ admits a right
spectral gap, quantitative bounds on the asymptotic variance of MCMC
estimates using $(\tilde{X}_{n})_{n\geq1}$ in relation to those using
$(X_{n})_{n\geq1}$ are available.

\subsection{Modification of a given factorisation} \label{sec:al_modification}

The easiest way to use the above result is to modify any candidate factorisation.
Given a factorisation of the function $r$
\[
r(x,y)=\prod_{k=1}^{d}\tilde{\rho}_{k}(x,y)\,,
\]
satisfying the balance condition,
we can define a sequence of functions $\rho_{k}$ such that both $r(x,y)=\prod_{k=1}^{d}\rho_{k}(x,y)$
and Condition~\ref{cond:rho_klower} holds. To that effect, take an arbitrary $c\in(0,1]$ and define $b:=c^{\frac{1}{d-1}}$.
Then, if we set
\[
\rho_{k}(x,y):=\min\left\{ \frac{1}{b},\max\left\{ b,\tilde{\rho}_{k}(x,y)\right\} \right\} ,\quad k\in\{1,\ldots,d-1\},
\]
it then follows that one must define 
\[
\rho_{d}(x,y):=\frac{r(x,y)}{\prod_{k=1}^{d-1}\rho_{k}(x,y)}.
\]
From this modification, we deduce the following result: 
\begin{prop}\label{prop:better}
Using this scheme, Lemma~\ref{lem:alv_dirichet_comparison} holds
with $\Pi_{1}=P$, $\Pi_{2}=\tilde{P}$, $\mu=\pi$ and $\varrho=c^{2}$.\end{prop}
\begin{proof}
We note that $\inf_{(x,y)\in\mathsf{X}^{2}}\prod_{k=1}^{d-1}1\wedge\rho_{k}(x,y)\geq b^{d-1}=c$
and that
\[
\tilde{\rho}_{d}(x,y)=\frac{r(x,y)}{\prod_{k=1}^{d-1}\rho_{k}(x,y)}\geq b^{d-1}r(x,y)=cr(x,y).
\]
With $A:=\{(x,y)\in\mathsf{X}^{2}:r(x,y)\geq1\}$, it follows that
$\inf_{(x,y)\in A}\tilde{\rho}_{d}(x,y)\geq c$, and so $\inf_{(x,y)\in A}\tilde{\alpha}(x,y)\geq c^{2}$.
We conclude along the same line as in the proof of Proposition~\ref{prop:generic_da_result}.
\end{proof}
While this modification ensures that one can take $\varrho=c^{2}$
in Proposition~\ref{prop:generic_da_result}, it is too general to
suggest that using $\tilde{P}$ can be more computationally efficient
than using $P$ when the cost of evaluating each $\rho_{k}$ is taken into
account. Indeed, Proposition \ref{prop:better} holds when the functions $\tilde{\rho}_{k}$
are chosen completely arbitrarily. Of course in practice, one should
choose $\tilde{\rho}_{k}$ and hence $\rho_{k}$ so that they are
in some sense in agreement with $r$.

We will show in the next example that a certain class of $\tilde{\rho}_{k}$'s
are beneficial, namely those which correspond to Metropolis--Hastings acceptance
ratios with ``flattened'' surrogate target densities. On the other
hand, it is far from difficult to come up with surrogate target densities
for which unmodified use of $\tilde{\rho}_{k}$ can lead to disastrous
performance.

\subsection{Example: unmodified surrogate targets}

One common usage \citep{christen:fox:2005} of Delayed Acceptance is to substitute a surrogate
target $\bar{\pi}$ for $\pi$ in $\rho_{1}(x,y)$. We consider the
case $d=2$ and a random walk proposal to examine Condition~\ref{cond:rho_klower}
in this context. Here we have $q(x,y)=q(y,x)$ and so 
\[
\alpha(x,y)=1\wedge\frac{\pi(y)}{\pi(x)},
\]
while
\[
\rho_{1}(x,y)=1\wedge\frac{\bar{\pi}(y)}{\bar{\pi}(x)},\quad\rho_{2}(x,y)=1\wedge\frac{\pi(y)\bar{\pi}(x)}{\pi(x)\bar{\pi}(y)}.
\]

Considering $(x,y)\in A=\{(x,y)\in\mathsf{X}^{2}:r(x,y)\geq1\}$ we
require $c>0$ satisfying simultaneously
\[
\frac{\bar{\pi}(y)}{\bar{\pi}(x)}\geq c,\qquad\frac{\pi(y)\bar{\pi}(x)}{\pi(x)\bar{\pi}(y)}\geq c.
\]

The first of these says that $\bar{\pi}(y)/\bar{\pi}(x)$ cannot be too small when $\pi(y)\geq\pi(x)$. The second says
that $\bar{\pi}(y)/\bar{\pi}(x)$ should not be a large multiple of $\pi(y)/\pi(x)$. There are a large variety of choices
of $\bar{\pi}$ that allow one to take $c=1$.  For example, $\bar{\pi}(x)=\pi(x)+C$ for some constant $C\geq0$ and
$\bar{\pi}(x)\propto\pi(x)^{\beta}$ for some $\beta\in[0,1]$.  Note that $\beta=0$ corresponds to $\bar{\pi}$ being a
constant function and $\beta=1$ corresponds to $\bar{\pi}\propto\pi$. In between, one can think of $\bar{\pi}$ as being
a flattened version of $\pi$.

\subsection{Counter-example: failure to reproduce geometric ergodicity}\label{sec:countreX}

Consider the case $\pi(x)=\mathcal{N}(x;0,1)$ and $\bar{\pi}(x)=\mathcal{N}(x;0,\sigma^{2})$
with $Q(x,\cdot)$ a normal distribution with mean $x$ and fixed
variance for each $x\in\mathbb{R}$. Here we have

\[
\rho_{1}(x,y)=\exp\left\{ \frac{(x-y)(x+y)}{2\sigma^{2}}\right\} ,\quad\rho_{2}(x,y)=\exp\left\{ \frac{(1-\sigma^{2})(y-x)(y+x)}{2\sigma^{2}}\right\} .
\]
\citet{mengersen:tweedie:1996} showed that a random-walk Metropolis--Hastings chain for targets with super-exponential
tails is geometrically ergodic. We now exploit this result to derive that, if $\sigma^{2}<1$, then the unmodified
delayed acceptance Markov chain is not geometrically ergodic.
\begin{prop}
The unmodified Delayed Acceptance Markov chain using the factorisation into $\rho_{1}$ and
$\rho_{2}$ as above is not geometrically ergodic when $\sigma<1$.\end{prop}
Intuitively, when $x$ is large $P(x,(-\infty,x))\approx\frac{1}{2}$ but $\lim_{x\rightarrow\infty}\tilde{P}(x,\{x\})=1$ because $\rho_{1}(x,y)$
takes on smaller and smaller values for $y>x$ and $\rho_{2}(x,y)$ takes on smaller and smaller values for $y<x$.
\begin{proof}
From \citet[Theorem~5.1]{roberts:tweedie:1996}, it suffices to show that $\pi{\rm-ess}\inf_{x\in\mathsf{X}}\tilde{P}(x,\{x\}^{\complement})=0$,
i.e. that for any $\tau\in(0,1)$ we can find $A\subseteq\mathsf{X}$
such that $\pi(A)>0$ and $\sup_{x\in A}\tilde{P}(x,\{x\}^{\complement})\leq\tau$.
Let $B_{s}(z)$ denote the ball of radius $s$ around $z$. Given
$\tau\in(0,1)$, we define 
\[
r:=\sup\{s>0:Q(x,B_{s}(x))<\tau/3\},
\]
and 
\[
A:=\left\{ x:x>\frac{r}{2}-\frac{\sigma^{2}\log(\tau/3)}{r(1-\sigma^{2})}\right\} \bigcap\left\{ x:Q(x,\mathbb{R}_{-})<\frac{\tau}{3}\right\} .
\]
Then 
\begin{eqnarray*}
\tilde{P}(x,\{x\}^{\complement}) & = & \tilde{P}(x,B_{r}(x)\setminus\{x\})+\tilde{P}(x,B_{r}^{\complement}(x))\\
 & \leq & \frac{\tau}{3}+\int_{B_{r}^{\complement}(x)}Q(x,{\rm d}y)\tilde{\alpha}(x,y)\\
 & \leq & \frac{2\tau}{3}+\int_{B_{r}^{\complement}(x)\cap\mathbb{R}_{+}}Q(x,{\rm d}y)\tilde{\alpha}(x,y)\\
 & \leq & \frac{2\tau}{3}+\sup_{y\in B_{r}^{\complement}(x)\cap\mathbb{R}_{+}}\tilde{\alpha}(x,y)\\
 & = & \frac{2\tau}{3}+\sup_{y\in B_{r}^{\complement}(x)\cap\mathbb{R}_{+}}\left[1\wedge\rho_{1}(x,y)\right]\left[1\wedge\rho_{2}(x,y)\right].
\end{eqnarray*}
Now let $x\in A$, $y\in B_{r}^{\complement}(x)\cap\mathbb{R}_{+}$
and assume $y<x$. It follows that $\rho_{2}(x,y)$ attains its maximum
when $y=x-r$ and therefore
\begin{eqnarray*}
\rho_{2}(x,y) & \leq & \exp\left\{ \frac{(1-\sigma^{2})r(r-2x)}{2\sigma^{2}}\right\} \\
 & \leq & \exp\left\{ \frac{(1-\sigma^{2})r}{2\sigma^{2}}\left[\frac{2\sigma^{2}\log(\tau/3)}{r(1-\sigma^{2})}\right]\right\} =\frac{\tau}{3}.
\end{eqnarray*}
Similarly, let $x\in A$, $y\in B_{r}^{\complement}(x)\cap\mathbb{R}_{+}$
and assume $y>x$. It follows that $\rho_{1}(x,y)$ attains its maximum
when $y=x+r$ and therefore
\begin{eqnarray*}
\rho_{1}(x,y) & \leq & \exp\left\{ -\frac{r(2x+r)}{2\sigma^{2}}\right\} \\
 & \leq & \exp\left\{ -\frac{r}{2\sigma^{2}}\left(2r-\frac{2\sigma^{2}\log(\tau/3)}{r(1-\sigma^{2})}\right)\right\} \\
 & \leq & \exp\left\{ \frac{r}{2\sigma^{2}}\left(\frac{2\sigma^{2}\log(\tau/3)}{r(1-\sigma^{2})}\right)\right\} \leq\frac{\tau}{3},
\end{eqnarray*}
since $\log(\tau/3)<0$ and $\sigma^{2}<1$. Therefore,
\[
\sup_{y\in B_{r}^{\complement}(x)\cap\mathbb{R}_{+}}\left[1\wedge\rho_{1}(x,y)\right]\left[1\wedge\rho_{2}(x,y)\right]\leq\frac{\tau}{3}
\]
so $\tilde{P}(x,\{x\}^{\complement})\leq\tau$ and we conclude.
\end{proof}
The same argument can be made for much more general targets and proposals, albeit at the expense of brevity and
clarity. We refrain from such a generalisation  as our purpose here is to demonstrate that the DA chain may fail to inherit geometric
ergodicity and that the simple proposed modification of the Delayed Acceptance kernel provided in Section
\ref{sec:al_modification} allows one to avoid this.

\section{Optimisation} \label{sec:optim}

When considering Markov Chain Monte Carlo methods in practice, their efficiency as measured by mixing properties and
computational cost is a fundamental issue. This section addresses both perspectives in connection with Delayed
Acceptance. Section \ref{sec:optim_proposal} examines the proposal distribution and derives its optimal scaling from standard random--walk Metropolis--Hastings theory. Then Section \ref{sec:optim_rank} covers the ranking of the factors $\rho_i$, which drives the total
computational cost of the procedure.

\subsection{Optimising the proposal mechanism} \label{sec:optim_proposal}

The explorative performances of a random--walk MCMC are strongly dependent on its proposal distribution. As exemplified in \cite{gilksetal97}, finding the optimal scale parameter does lead to efficient `jumps' in the state space and the overall acceptance rate of the chain is directly connected to the average jump distance and to the asymptotic variance of ergodic averages. This provides practitioners with an approach to `auto-tune' the resulting random--walk MCMC algorithm. Extending this calibration to the Delayed Acceptance scheme is equally important, on its own ground towards finding a reasonable scaling for the proposal distribution and to avoid comparisons with the standard Metropolis--Hastings version.

The original framework of \cite{gilksetal97} is cantered on estimating a collection of expected functionals, say $g$,
where a plausible criterion for the performances of the MCMC is the minimisation of the stationary integrated
auto-correlation time (ACT) of the Markov chain, defined as
$$ \tau_g = 1 + 2\sum \limits_{i=1}^\infty \mathbb C\text{or}( g(X_0), g(X_i) ) $$
where the index $g$ stresses the dependence on the considered functional, 
which is connected to the asymptotic variance through $\varm(P,g) = \tau_g \times \varm_\pi(g)$ whenever the chain is
$\phi$-irreducible, aperiodic, and reversible, $\varm_\pi(g)$ is finite and $g \in L^2(\pi)$.

Research on this optimisation focus on two main cases:
\begin{itemize}
\item Consider the limit in the dimension of the target distribution toward $\infty$, where \cite{gilksetal97} gave
conditions under which each marginal chain converges toward a Langevin diffusion. Maximising the speed of that
diffusion, say $h(\ell)$ where $\ell$ is a parameter of the scale of the proposal, implies a minimisation of the ACT 
and also that $\tau$ is \emph{free} from the dependence on the functional, defining thus an independent measure of efficiency for the algorithm;
\item \cite{sherlock:09} focus on unimodal elliptically symmetric targets and show that a proxy for the ACT in finite
dimensions is the Expected Square Jumping Distance (ESJD), defined as 
$$ \mathbb E \left[ \| X' - X \|_\beta^2 \right] = \mathbb E \left[ \sum \limits_{i=1}^d \frac{1}{\beta_i^2} (X'_i - X)^2 \right]$$
where $X$ and $X'$ are two successive points in the chain and $\|\cdot\|_\beta$ represent the norm on the
principal axes of the ellipse rescaled by the coefficients $\beta_i$ so that every direction contributes equally.
\end{itemize}

An interesting result in \cite{sherlock:09} is that, as $d \rightarrow \infty$, the ESJD on one marginal component of
the chain converges with the same speed as the diffusion process described in \cite{gilksetal97}. 
These authors furthermore show the asymptotic result holds for rather small dimension, roughly starting from $d=5$.

Moreover, when considering efficiency for Delayed Acceptance, which is a technique tailored on costly computations, we need to focus on the execution time of the algorithm as well.
We then proceed to define our measure of efficiency as 
\begin{equation}\label{eq:Eff}
\mathbf{Eff} := \text{ESJD}\big/{\text{cost per iteration}}
\end{equation}
similarly to \cite{sherlock:13} for Pseudo-Marginal MCMC.

Due to the complex acceptance ratio in Delayed Acceptance, an extension of the previous results requires rather
stringent assumptions, albeit providing a proper guideline in practice. Section \ref{sec:examples} will further
demonstrate optimality extends beyond those conditions. Note that our assumptions are quite standard in the literature
on the subject. \\

\noindent (H1) We assume for simplicity's sake that the Delayed Acceptance procedure operates on two factors only, i.e., that $ r(x,y)
= \rho_1(x,y) \times \rho_2(x,y) $.
The acceptance probability of the scheme is thus $$\tilde \alpha(x,y) = \prod \limits_{i=1}^2 (1 \land \rho_i(x,y)) 
.$$
We also consider the ideal setting where a computationally cheap approximation $\tilde f(\cdot)$ is available for
$\pi(\cdot)$ and precise enough so that $\rho_2(x,y) = r(x,y) / \rho_1(x,y) = {\pi(y)}\big/{\pi(x)} \times {\tilde
f(x)}\big/{\tilde f(y)} = 1$.\\

\noindent (H2) We further assume that the target distribution satisfies (A1) and (A2) in \cite{gilksetal97}, which are regularity
conditions on $\pi$ and its first and second derivatives, and that $\pi(x) = \prod \limits_{i=1}^n f(x_i)$. 

\noindent (H3) We consider a random walk proposal $y = x + \sqrt{ \ell^2 / d } \, Z$, where $Z \sim \mathcal N(0, I_d)$. Note that Gaussianity can be easily relaxed to distributions with finite fourth moment and similar results are available for more
heavy-tailed distributions \citep{pneal:11}.\\

\noindent (H4) Finally we assume that the cost of computing $\tilde f(\cdot)$, say $c$, is proportional to the cost of computing $\pi(\cdot)$, named $C$, 
with $c = \delta C$.\\

Normalising costs by setting $C=1$, the average total cost per iteration of the Delayed Acceptance chain is 
$ \delta + \mathbb E \left[ \tilde \alpha \right]$ and the efficiency of the proposed method under the above conditions is
$$ \mathbf{Eff} (\delta,\ell) = \frac{ESJD}{\delta + \mathbb E \left[ \tilde \alpha \right]}$$

\begin{lem} \label{lem:optim} 
Under the above conditions (H1--H4) on the target $\pi(x)$, on the proposal $q(x,y)$ and on the factorised acceptance probability $\tilde \alpha(x,y) = \prod \limits_{i=1}^2 (1 \land \rho_i(x,y))$ we have that
$$ \tilde \alpha(x,y) =  \left(1 \land \rho_1(x,y) \right)$$
 and that as $d \rightarrow \infty$
 $$ \mathbf{Eff} (\delta,\ell) = \frac{h(\ell)}{\delta + \mathbb E \left[ \tilde \alpha \right]} = \frac{2 \ell^2
\Phi(-\frac{\ell \sqrt{I}}{2})}{\delta + 2\Phi(-\frac{\ell \sqrt{I}}{2})}$$  
 $$a(l) = \mathbb E \left[ \tilde \alpha \right] = 2 \Phi(-\frac{\ell \sqrt{I}}{2})$$
 where $I := \mathbb E \left[ \left( \frac{ ( \, \pi(x) \,)'}{ \pi(x)} \right)^2 \right]$ as defined in \cite{gilksetal97}.

\end{lem}
\begin{proof}
It is easy to see that (H1) implies $ \tilde f(\cdot) = \pi(\cdot) $ and so $\rho_1(x,y) = r(x,y)$. Moreover, by definition,
$\rho_2(x,y) = r(x,y)/\rho_1(x,y) =1 $ and hence the second test is always accepted. The acceptance rate reduces then to
just the ratio $ \tilde f(y) / \tilde f(x) = \rho_1(x,y)$.

The second part of the lemma follows directly from Theorem 1.1 in \cite{gilksetal97}.
\end{proof}

Let us stress that almost all assumptions in the above Lemma can be relaxed and that performances are robust against
small deviances from those assumptions, as shown by the literature on standard Metropolis--Hastings. Obtaining
analytical results without such conditions, while possible, requires however a considerable mathematical effort.

We now state the main practical implication of Lemma \ref{lem:optim}.
\begin{prop}
If the conditions of Lemma \ref{lem:optim} holds, the optimal average acceptance rate $\alpha^*(\delta)$ is independent of $I$.
\end{prop}

\begin{proof}
Consider  $\mathbf{Eff} (\delta,\ell)$ in terms of $(\delta,a(\ell))$: $$ a = g(\ell) = 2 \Phi \left(-\frac{\ell
\sqrt{I}}{2} \right) \quad ; \quad \ell = g^{-1}(a) = -\Phi^{-1}\left(\frac{a}{2}\right)\frac{2}{\sqrt{I}}$$
$$
\mathbf{Eff} (\delta,a) = \frac{ \frac{4}{I} \left[ \Phi^{-1} \left( \frac{a}{2} \right)^2 a \right] } {\delta + a} =  \frac{4}{I} \left\{ \frac{1}{\delta + a} \left[ \Phi^{-1} \left( \frac{a}{2} \right)^2 a \right] \right\}
$$
where we dropped the dependence on $\ell$ in $a$ for notation's sake.
It is now evident that to maximise $\mathbf{Eff} (\delta,a)$ in $a$ we only need maximise $\left\{
\frac{1}{\delta + a} \left[ \Phi^{-1} \left( \frac{a}{2} \right)^2 a \right] \right\}$, which is independent of $I$.
\end{proof}

\begin{figure}
\caption{Two top panels: behaviour of $\ell^*(\delta)$ and $\alpha^*(\delta)$ as the relative cost varies. Note that for
$\delta >> 1$ the optimal values converges towards the values computed for the standard Metropolis--Hastings (dashed in
red).  Two bottom panels: close--up of the interesting region for $0 < \delta < 1$.}
\label{fig:alpha_l_delta}
\includegraphics[scale=0.45]{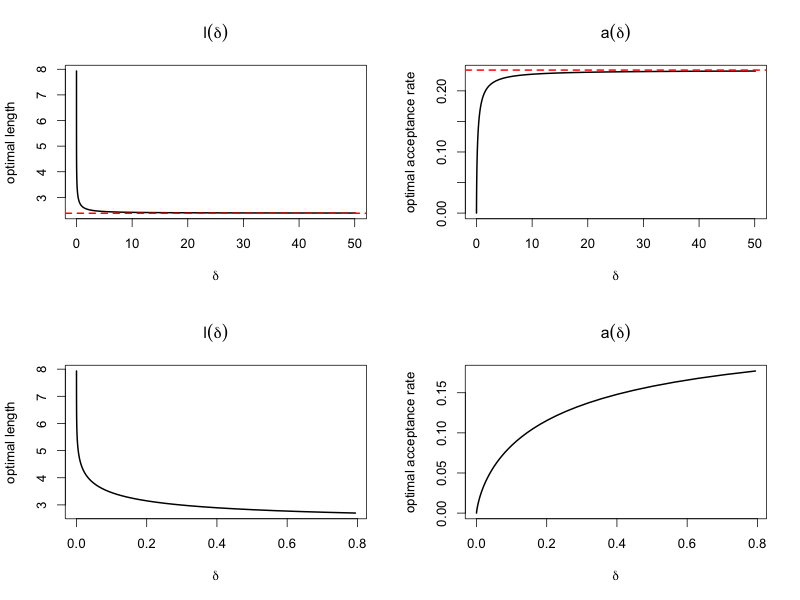}
\end{figure}

The optimal scale of the proposal $\ell^*(\delta)$ and the optimal acceptance rate $a^*(\delta)$ are thus given as
functions of $\delta$. In particular, as the relative cost of computing $\rho_1(x,y)$ with respect to $\rho_2(x,y)$
decreases, the proposed moves become bolder, in that $\ell^*$ increases and $a^*$ decreases, since rejecting costs
the algorithm  little in terms of time, while every accepted move results in an almost independent sample. On the
contrary when $\delta$ grows larger the chain rapidly approaches a Metropolis--Hastings behaviour, as it is no
longer convenient to reject early. Figure \ref{fig:alpha_l_delta} helps visualise the result. \\

\subsection{Ranking the Blocks}\label{sec:optim_rank}

As mentioned at the end of Section \ref{sec:zero_zero}, the order in which the factors $\rho_i(x,y)$ are tested has
a strong influence on the performance of the algorithm. Delayed Acceptance was first developed in
\cite{fox:nicholls:1997,christen:fox:2005} to speed up computations using a cheap approximation $\tilde \pi(\cdot)$ of
the target distribution $\pi(\cdot)$ as a first step before computing the actual, and costly, Metropolis--Hastings ratio
${\pi(y)}\big/{\pi(x)}$ only in the cases where the acceptance test based on the approximation $\tilde \pi$ was satisfied.
The main idea, namely to avoid the computation of the most costly parts as often as possible, remains relevant even for factorisations composed of more than two terms.

Consider an i.i.d. framework; the target (in $x$) is given by 
$$
\pi(x|Z) \propto p(x) \times L(Z|x) = p(x|Z) \times \prod \limits_{i=1}^n f(z_i|x) 
$$ 
where $Z=(z_1, \dots, z_n)$ is an i.i.d. sample from $f(z|x)$ and
$p(x)$ is the prior distribution for $x$, we can always consider the decomposition
\begin{equation}\label{eq:xi_factors} 
r(x,y) = \prod \limits_{i=1}^K \xi_i(x,y) \end{equation}
where each $\xi_i(x,y)$ is made of a small number of density ratio terms, with one including the prior and
proposal ratios. In the limit, it is feasible if not necessarily efficient to consider the case $K=n+1$ with
$$
\xi_i(x,y) = \frac{f(z_i|y)}{f(z_i|x)} \quad i = 1, \dots \,, 
n\ \text{ and }\ \xi_{n+1}(x,y) = \frac{p(y)q(y,x)}{p(x)q(x,y)}.
$$

Assuming the computing cost is comparable for all terms, a solution for optimising the order of these factors ranks the entries according to the success rates observed so far, starting with the least successful values. Alternatively, the factorisation can start with the ratio that has the highest variance, since it is the most likely to be rejected. (Note however that poor factorisations \eqref{eq:xi_factors} lead to very low acceptance rates, as for instance when picking only outliers in a given group of observations.) 
Lastly, we can rank factors by their correlation with the full Metropolis--Hastings ratio; taking the argument to the
limit, if the first factor has a perfect correlation with $r(\cdot,\cdot)$ then all the successive terms must be accepted and their order is hence of no interest.

 
This later setting is akin to considering the hypothetical optimal solution introduced in Section \ref{sec:optim_proposal} with only two terms in the decomposition. Let a small number of the best scoring terms be merged to form $\rho_1$ and let the remaining factors become $\rho_2$. $ \rho_1(x,y) = \tilde \pi(y) / \tilde \pi(x)$ is then highly correlated to $r(x,y)$, $ \rho_2(x,y) \sim 1 $ for every $(x,y)$ and hence $\tilde \pi(x)$ is a close approximation of the target, albeit probably flattened, which is exactly what we want (see Section \ref{sec:convergence}).
 
As all these features can be evaluated for each subsample while running a chain with acceptance ratio factored as in \eqref{eq:xi_factors}, an implementation based on this intuition is then to take
$$
\tilde \pi_{Z^*}(x) \propto p(x)^{m/n} \prod \limits_{i=1}^{m} f(z^*_i|x)
 \quad\text{ or }\quad \tilde \pi_{Z^*}(x) \propto \prod \limits_{i=1}^{m} f(z^*_i|x)
$$
with $m < n$, where $Z^* = (z^*_1, \dots, z^*_{m})$ is a subsample of $Z$.
At each iteration $t$ of the Markov chain we compute all the $\xi_i(x^{(t-1)},y)$ and $Z^{*(t)}$ is chosen as the subset
that maximise the observed correlation between the values of $\frac{ \tilde \pi_{Z^*}(x^{(j)})}{\tilde
\pi_{Z^*}(x^{(j)})} \quad (j=1,\dots \,, t-1)$ and the full Metropolis--Hastings ratio (or whatever other selected
criterion). As computing the real $\argmax_{ Z^* \subset Z }$ is expensive, in our practical implementation we resort to a
\emph{forward} selection scheme; starting with the factor $\xi_i$ with the maximal correlation we build $Z^{*(t)}$
merging one term at a time until a desired correlation level is achieved, the observed correlation after including
another term does not grow more than a small $\epsilon>0$ or the size of $Z^{*(t)}$ has reached a critical point for
computational purposes (e.g. $10\%$ of the whole sample $Z$).

A relevant warning is that if we rearrange terms during the run, not only reordering but also merging them, in
accordance to their correlation with the unmodified ratio, the resulting method has no theoretical guarantee since the
kernel is potentially changing at each iteration depending on properties of previous samples \citep{gelfand:sahu:1994}.

As with standard adaptive MCMC \citep{roberts:rosenthal:2005} we resort thus to a \emph{finite adaptation scheme};
we start with a fixed number of iterations to rank and rearrange the factors, followed by a fixed ordering to achieve ergodicity of the chain.
We test this procedure in Section \ref{sec:ex_logistic} on a simulated example.

Finally note that while we focused on the i.i.d.~setting, in more complex cases where the ratio is factored and Delayed
Acceptance can be applied, it is often the case that the optimal ordering of such factors is already known.


\section{Relation with other methods} \label{sec:others}

\subsection{Delayed Acceptance and Prefetching} \label{sec:optim_prefetch}

Prefetching, as defined by \cite{brockwell:2006}, is a programming method that accelerates the convergence of a single
MCMC chain by guessing future states in the path of a random walk Metropolis--Hastings Markov chain in order to use any
additional computing power available, in the form of extra parallel processors, to calculate in advance necessary
quantities (like the Metropolis--Hastings ratio) so that when the chain reaches a given state the computationally-heavy
part of that iteration are ready.

Clearly the usefulness of this technique depends on our ability to guess the path of the chain correctly and hence many
advanced prefetching strategies make use of the observed acceptance rate of the chain or even of a fast approximation
$\tilde \pi$ of the target distribution to select the most likely future outcomes.

Since an in-depth exploration of prefetching is outside the scope of this work the reader is referred to
\cite{strid:2010} and citations therein for a complete discussion of the argument. 

As mentioned above and demonstrated in \cite{strid:2010,angelino:etal:2014} if a cheap approximation $\tilde{\pi}$ of
the target density is available, it can be used to select more likely future paths of the chain and this results in an
efficient prefetching algorithm.

In our case the master process sequentially samples from the (Delayed Acceptance) chain by checking only the
(assumed) inexpensive first approximation $\rho_1(x,y) = \tilde \pi(y) / \tilde \pi(x)$ while the other additional
processors provide him the more expensive $\rho_2(x,y) = \pi(y)\tilde \pi(x) / \pi(x)\tilde \pi(y)$ computed
beforehand thanks to prefetching. The theoretical properties of the chain are unchanged while the achievable
speed-up may be substantial, especially for the first few additional processors.

\subsection{Alternative procedure for Delayed Acceptance} \label{sec:phil}

In the case that every factor $\rho_i(x,y)$ has roughly the same computational cost, Philip Nutzman suggested (personal
communication) that Delayed Acceptance can be slightly modified by taking the overall acceptance probability 
$$
\prod_{i=1}^d \min\left\{\rho_i(x,y),1\right\} \quad \text{ to be instead } \quad \min \limits_{k=1,\dots\,,d} \left\{ \prod_{i=1}^k
\rho_i(x,y),1\right\} \,.
$$
Such a decomposition follows from the same idea that one would like to compute as few factors as possible once one
realizes that the proposal is likely to be rejected. Under this modification the associated Markov chain still achieves the correct
target in the stationary regime and the procedure satisfies detailed balance, provided the ordering of the terms is
uniformly random.

An interesting consequence of this modification is that, as the number of factor increases, the acceptance rate
eventually stabilises, while for the method described in Section \ref{sec:zero_zero} the acceptance rate decreases to
zero. This property is indeed appealing, even thought this procedure logically takes longer to complete when compared
with the standard Delayed Acceptance (albeit less than the reference Metropolis--Hastings procedure).

The evident disadvantage of the modification in a general setting is that detailed balance implies that the factors are computed in
a random order at each iteration, making vain any attempt to adapt in terms of the ordering (Section
\ref{sec:optim_rank}) or to set the order based on respective computational costs.

This drawback can be somewhat reduced by combining the above two approaches; consider the decomposition
$$ {\pi(y)\,q(y,x)}\big/{\pi(x)q(x,y)} = \left[ \prod_{i=1}^{d_1} \rho_i(x,y) \right] \times \left[ \prod_{j=1}^{d_2} \phi_j(x,y) \right]$$
where $d_1 + d_2 = d$ and the factors $\rho_i$ and $\phi_j$ represent respectively cheap factors and costly factors.
By taking now 
$$
\min \limits_{(m=1,\dots\,,d_1)} \left\{1,\prod \limits_{i=1}^m{\rho_i(x,y)} \right\} \times \min
\limits_{(k=1,\dots\,,d_2)} \left\{1,\prod \limits_{j=1}^k {\phi_i(x,y)} \right\} 
$$
the algorithm computes cheap factors first and expensive factors last, applying the symmetry requirement to satisfy
detail balance inside each of both subsets. Clearly the above can be generalised to a larger number of subsets, each with $d_i$ factors in it.
Intuitively, this last modification can be explained as an early rejection of each of the intermediate acceptance/rejection steps inside a Delayed 
Acceptance scheme.

\begin{remark} Interestingly if $d_l  = 1 \ \forall l$ ($l$ being the number of subsets considered) this procedure
reduces to Delayed Acceptance, and for $l$ that increases and $d_l  > 1 \ \forall l$ this combined technique will
have a even lower overall acceptance rate than standard Delayed Acceptance. \end{remark}

\subsection{Delayed Acceptance and Slice Sampling}

As a final remark, we stress another analogy between our Delayed Acceptance algorithm
and slice sampling \citep{neal:1997,robert:casella:2004}.  Based on the same decomposition 
\eqref{eq:baselike}, slice sampling proceeds as follows
\begin{enumerate}
\item simulate $u_1,\ldots,u_n\sim\mathcal{U}(0,1)$ 
and set $\lambda_i=u_i\ell(\theta|x_i)$ $(i=1,\ldots,n)$;
\item simulate $\theta^\prime$ as a uniform under the constraints
$\ell_i(\theta^\prime|x_i)\ge \lambda_i$ $(i=1,\ldots,n)$.
\end{enumerate}
to compare with Delayed Acceptance which conversely
\begin{enumerate}
\item simulate $\theta^\prime\sim q(\theta^\prime|\theta)$;
\item simulate $u_1,\ldots,u_n\sim\mathcal{U}(0,1)$
and set $\lambda_i=u_i\ell(\theta|x_i)$ $(i=1,\ldots,n)$;
\item check that
$\ell_i(\theta^\prime|x_i)\ge \lambda_i$ $(i=1,\ldots,n)$.
\end{enumerate}
The differences between both schemes are thus that (a) slice sampling always accepts a move,
(b) slice sampling requires the simulation of $\theta^\prime$ under the constraints, which may
prove infeasible, and (c) Delayed Acceptance re-simulates the uniform variates in the event of a rejection. In this
respect, Delayed Acceptance appears as a ``poor man's" slice sampler in that values of $\theta's$ are proposed until one
is accepted.

\section{Examples}\label{sec:examples}

To illustrate the improvement brought by Delayed Acceptance, we study three different realistic settings that reflect on
the generality of the method. First, in Section \ref{sec:ex_logistic} we consider a Bayesian analysis of a logistic
regression model, to assess the computational gain brought by our approach in a ``Big-Data" environment where obtaining
the likelihood is the main computational burden.  Secondly (Section \ref{sec:ex_mala}) we examine a high dimensional toy
Normal-Normal model, sample with a geometric Metropolis adjusted Langevin algorithm where the main computational cost
comes from the proposal distribution which is position specific and involves derivatives of the density up till the
third level, which are computed numerically at each iteration.  Finally in Section \ref{sec:mixis} we investigate a
mixture model where a formal Jeffreys prior is used, as it is not available in closed-form and does require an expensive
approximation by numerical or Monte Carlo means. The latter example comes as a realistic setting where the prior itself
is a burdensome object, even for small datasets.

\subsection{Logistic Regression} \label{sec:ex_logistic}

While a simple model, or due to its simplicity, logistic regression is
widely used in applied statistics, especially in classification problems.  The
challenge in the Bayesian analysis of this model is not generic, since simple
Markov Chain Monte Carlo techniques providing satisfactory approximations, but
stems from the data-size itself. This explains why this model is used as a
benchmark in some of the recent accelerating papers
\citep{korattikara:chen:welling:2013, neiswanger:wang:xing:2013,
scott:etal:2013,wang:dunson:2013}. Indeed, in ``big Data" setups, MCMC is
deemed to be progressively inefficient and researchers are striving to keep
simulation effective, focusing mainly on parallel computing and on sub-sampling but also on replacing the classic
Metropolis scheme itself. 

We tested the proposed method against the standard Metropolis--Hastings algorithm on $10^6$ simulated data with a
$100$-dimensional parameter space.  The proposal distribution is Gaussian: $ y|x \sim \mathcal N(x,\Sigma)$ with
$\Sigma$ initialised to be $0.2 \times I_d$ ($d$ being the dimension of the parameter space) and then adapted.  The
Metropolis--Hastings benchmark was made adaptive by targeting the asymptotic optimal acceptance rate of $\alpha^*=0,234$
\citep{gilksetal97}.

Delayed Acceptance was optimised first against the ordering of the factors as explained in Section \ref{sec:optim}; we
split the data into subsamples of $10$ elements and computed their empirical correlation with the full
Metropolis--Hastings ratio as a criterion. Once these estimates were stable we merged into the surrogate target $\tilde
f$ the smallest number of subsamples needed to achieve a $\ge 0,85$ correlation with $r(x,y)$. As soon as the ordering
was fixed we computed $\delta$, the relative cost of the obtained $\rho_1$ with respect to the full ratio, and run the
chain for the remaining iterations optimising $\Sigma$ against the optimal acceptance rate found through \eqref{eq:Eff}.
We also added the modification explained in Section \ref{sec:al_modification} with $c$ set such that $b$ was slightly
lower than the optimal acceptance rate above. 

\begin{table}
    \begin{tabular}{ | l | c | c | c |}
    \hline
    \textbf{Algorithm} & \textbf{relative ESS} (aver.) & \textbf{relative ESJD} (aver.) & \textbf{relative Time} (aver.) \\ \hline
     DA-MH over MH & 1.1066 & 12.962 & 0.098 \\ \hline
    \end{tabular}
   
   \begin{center} 
   \begin{tabular}{ | l | c | c |}
    \hline
    \textbf{Algorithm} & \textbf{relative Eff gain (ESS)} (aver.) & \textbf{relative Eff gain (ESJD)} (aver.) \\ \hline
     DA-MH over MH & 5.47 & 56.18 \\ \hline
    \end{tabular}
   \end{center}
\caption{ Comparison between MH and MH with Delayed Acceptance on a logistic model. \textbf{ ESS} is the effective sample size, \textbf{ ESJD} the expected square jumping distance, \textbf{time} is the computation time.} \label{tab:logistic}
\end{table}

We collected $100$ repetitions of the experiment and the results are presented in Table \ref{tab:logistic}.
Before commenting the results we highlight the fact that this situation may seem not particularly appealing for Delayed
Acceptance and in fact straight application of the method by randomly choosing the composition of $\rho_1$ and $\rho_2$
may lead to variable results. Further coding effort is required here in order to choose adaptively how to split
the MH ratio. Borrowing from both Section \ref{sec:optim_rank} and the end of Section \ref{sec:optim_proposal}, i.e. by
choosing during the brief burn-in of the chain which subset best represents the whole likelihood and then, based on how
populated that subset is, targeting a specific acceptance ratio, produces both a completely automated MCMC version for
this kind of data (\emph{iid}) and better results under a time constraint.

As shown in Table \ref{tab:logistic}, while the assumption made in Section \ref{sec:optim} not
completely satisfied, the relative efficiency of Delayed Acceptance is higher that for MH by a factor of almost $6$.
We measured efficiency trough \emph{effective sample size} (ESS, from the \textbf{coda} R package \citep{r:coda}) or \emph{expected square jumping distance} (ESJD).
By choosing the first subsample to be \emph{informative} on the whole ratio there is practically no loss on ESS (while the estimated ESJD
actually increased) and, given the significantly reduced acceptance rate, the computing time is usually less then a fourth of the computing time of the corresponding optimal MH, taking into account the first part of chain used to determine the blocks ranking.

\subsection{G-MALA with Delayed Acceptance} \label{sec:ex_mala}

\subsubsection{MALA and Geometric MALA:   }

Random walk Metropolis--Hastings, while generic and popular, can struggle with posterior distributions in high
dimensions or in the presence of high correlation between some components. In such cases it is inefficient, with
low acceptance rate, poor mixing and highly correlated samples.  Metropolis adjusted Langevin algorithm  (MALA, see for
instance \cite{roberts:stramer:2002}) has been devised to overcome these difficulties by taking advantage of 
the gradient of the target distribution in the proposal mechanism, making the Markov chain more robust with respect to
the dimension of the problem and proposing broader moves with higher probability.  MALA is based on a Langevin
diffusion, with the target (the posterior distribution $\pi(\theta|y)$ in our case) as a stationary distribution,
defined by the SDE 
$$
\frac{d \theta}{dt} = \nabla_\theta \log( \pi(\theta|y) ) \frac{dt}{2} + \frac{dB}{dt} 
$$
where $B$ is a Brownian motion. Using a first-order discretisation the diffusion gives the following proposal mechanism:
$$ 
\theta' = \theta^{(i-1)} +  \varepsilon^2  \nabla_\theta \log( \pi(\theta^{(i-1)}|y) ) / 2 + \varepsilon Z
$$
where $\varepsilon$ is the step-size for the Euler's integration and $Z \sim \mathcal{N}(\mathbf{0},\mathbf{I})$. This
discretisation is then compensated by introducing an accept/reject probability similar to a Metropolis--Hastings
algorithm.

This diffusion is isotropic and will hence still be inefficient for highly correlated components or with very different
scales, as the step size $\varepsilon$ is fixed across dimensions.
\cite{roberts:stramer:2002} propose to alleviate the issue using a pre-conditioning matrix $A$ so that the proposal becomes
$$ \theta' = \theta^{(i-1)} +  \varepsilon^2 A^TA  \nabla_\theta \log( \pi(\theta^{(i-1)}|y) ) / 2 + \varepsilon  A Z.$$
\cite{christensen:05} demonstrate however that defining this matrix in general can be difficult and that tuning on the
go may result in an inappropriate asymptotic behaviour.

In a recent work \cite{girolami:2011} propose the Geometric-MALA in order to overcome this difficulty, advising the use
of a position specific metric for the matrix $A$, which takes advantage of the geometry of the target space that the
chain is exploring. They suggest in particular the Fisher-Rao metric tensor.  In terms of Bayesian inference, where the
target distribution is the posterior density, this choice translates into $A^TA$ being the expected Fisher information
matrix plus the negative Hessian of the log-prior.

This theoretically efficient solution also performs well in practice but comes with a serious computational burden in
the fact that at every evaluation of the Metropolis--Hastings ratio derivatives up till the third order of our
log-target distribution are needed and, in the event of them being analytically not available, expensive numerical
approximations are to be computed (see equation (10) of \cite{girolami:2011}).

\subsubsection{Sampling with Delayed Acceptance and GMALA:   }

Geometric-MALA represent a perfect application for Delayed Acceptance since we can naturally divide its
acceptance ratio into the product of the posterior ratio and the ratio of the proposals, the latter to be only computed
when the proposed point is associated with a relatively large posterior probability.

As described above, the computational bottleneck of the G-MALA lays in the computation of the third derivative of our
log-target at the proposed point, while the computation of the posterior itself has usually a low relative cost.
Moreover even with a non-symmetric efficient proposal mechanism (the discretised Langevin diffusion) G-MALA is still
close to a random walk and we expect the ratio of the proposal to be near 1, especially at equilibrium especially when
$\varepsilon$ is small. Therefore, the first ratio is inexpensive, relative to the second one, while the decision reached 
at the first stage should be consistent with the overall acceptance rate.

Given that optimal scaling for MALA in terms of the dimension $d$ of the target differs from the
random-walk setting (see \citealp{roberts:rosenthal:2001}), we set the variance of the random-walk normal component 
as $\sigma^2_d = \frac{\ell^2}{d^{1/3}}$. Borrowing from Section \ref{sec:optim_proposal}, we can obtain
the optimal acceptance rate for the DA-MALA, through Equation \eqref{eq:Eff}, by maximising
$$ \mathbf{Eff} (\delta,\ell) = \frac{h(\ell)}{\delta + \mathbb E \left[ \tilde \alpha \right] - \delta \times \mathbb E
\left[ \tilde \alpha \right]} = \frac{ 2 \ell^2 \Phi(-\frac{ K \ell^3}{2}) }{\delta + 2\Phi(-\frac{ K \ell^3}{2}) \times
(1 - \delta) } $$  
or equivalently
$$
\mathbf{Eff} (\delta,a) = -\left(\frac{2}{K}\right)^{\frac{2}{3}} \left[ \frac{a \Phi^{-1} \left( \frac{a}{2}
\right)^{\frac{2}{3}}}{\delta + a (1 - \delta) } \right]\,.
$$
In the above the computational cost per iteration is taken to be $c=\delta C$ for the posterior ratio, $C=1$ for the proposal ratio (and hence $c + \mathbb E \left[ \tilde \alpha \right] (C-c)$ for the whole kernel), $h(\ell)$ is again the speed of the limiting diffusion process and $K$ is a measure of roughness of the
target distribution, depending on its derivatives. Since the optimal $a^*$ is independent from $K$, we do not define it
more rigorously, referring to \cite{roberts:rosenthal:2001}. Figure \ref{fig:MALA_optim} shows that $a^*$ decreases with $\delta$, 
as is the case with random-walk Metropolis--Hastings. It reaches the known optimum for the standard MALA when $\delta =1$.

\begin{figure}
\caption{Optimal acceptance rate for the DA-MALA algorithm as a function of $\delta$. In red, the optimal
acceptance rate for MALA obtained by \cite{roberts:rosenthal:2001} is met for $\delta=1$.}
\label{fig:MALA_optim}
\includegraphics[scale=0.45]{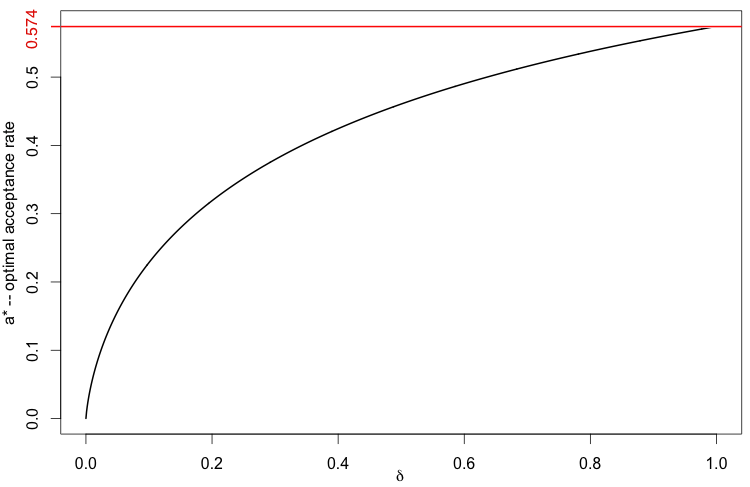}
\end{figure}

\subsubsection{Simulation study: }

To test the above assumptions we ran a toy MALA example where we drew $100$ samples from a $\mathcal N_d(\theta,I)$,
with $d=10$; $\pi(\theta)$ was set to be $\mathcal N_d(0,100)$. Figure \ref{fig:MALA_example} presents an example run.
We then repeated the experiment $100$ times and computed an average efficiency gain, defined either as ESS or as the ESJD, over the computing time.  We computed $\delta$ at each
run by averaging a few computed derivatives, required by the proposal ratio. We then adapt $\varepsilon$ to get the
optimal acceptance rate, being conservative in order to avoid overflow issues with the first-order numerical integrator.
Results are presented in Table \ref{tab:MALA}. Delayed Acceptance exhibits improvement by a factor of $10$ in this
example, obtained almost for free in terms of to coding time.

\begin{table}
\begin{center}
    \begin{tabular}{ | l | c | c | c | c | c | c |}
    \hline
    \textbf{Algorithm} & \textbf{ESS} (aver.) & \textbf{ESS} (sd) & \textbf{ESJD} (aver.) & \textbf{ESJD} (sd) & \textbf{time} (aver.) & \textbf{time} (sd) \\ \hline
       MALA & 7504.486 & 107.21 & 5244.946 & 983.473 & 176078 & 1562.3 \\ \hline
DA-MALA & 6081.023 & 121.42 & 5373.253 & 2148.761 & 17342.91 & 6688.3 \\ \hline    
    \end{tabular}
    
\vspace{5mm}    
 \begin{tabular}{ | l | c | c | c |}
    \hline
    \textbf{Algorithm} & {\textbf a} (aver.) & \textbf{ESS/time (aver.)} & \textbf{ESJD/time (aver.)} \\ \hline    
    MALA & 0.661 & 0.04 & 0.03 \\ \hline    
    DA-MALA & 0.09 & 0.35 & 0.31 \\ \hline    
  \end{tabular}
    
\end{center}
    \caption{ Comparison between standard geometric MALA and geometric MALA with Delayed Acceptance, with \textbf{ ESS} the effective sample size, \textbf{ ESJD } the expected square jumping distance, \textbf{time} the computation time and {\textbf a} the observed acceptance rate.}
\label{tab:MALA}
\end{table}

\begin{figure}
\caption{Comparison between geometric MALA (top panels) and geometric MALA with Delayed Acceptance (bottom panels):
marginal chains for two arbitrary components (left), estimated marginal posterior density for an arbitrary component
(middle), 1D chain trace evaluating mixing (right).}
\label{fig:MALA_example}
\includegraphics[scale=0.45]{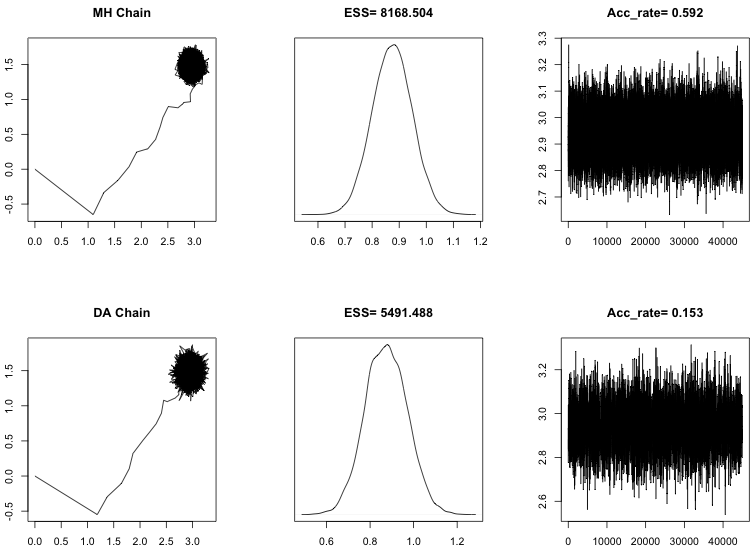}
\end{figure}

\subsubsection{HMC with Delayed Acceptance:   }

As a side note, while the reasoning applied to MALA does theory apply to Hamiltonian Monte Carlo (HMC), the
computational gain obtained through Delayed Acceptance is only connected with avoiding some proposal computations.
In a general HMC though (with both point-dependent and independent pre-conditioning matrices), proposing a new value
still involves the computation of $L-1$ (with $L$ the number of steps in the discretised--Hamiltonian
integration) derivatives, as only the starting point is recovered from the previous iteration, while computing the final
Metropolis--Hastings ratio involves just the extra computation at the end point. Therefore, in this setting, the
computational gain is much reduced.

\subsection{Mixture Model}\label{sec:mixis}

\subsubsection{Non-Informative inference on a Mixture Model: }
Consider a standard mixture model \citep{maclachlan:peel:2000} with a fixed number of components
\begin{equation}\label{eq:theMix}
\sum_{i=1}^k w_i\,f(x|\theta_i)\,,\quad\text{ with }\quad \sum_{i=1}^k w_i=1\,.
\end{equation}
This standard setting nonetheless offers a computational challenge in that the reference objective Bayesian approach 
based on the Fisher information and the associated Jeffreys prior \citep{jeffreys:1939,robert:2001}
is not readily available for computational reasons and has thus not been implemented so far. Proxys
using Jeffreys priors on the components of \eqref{eq:theMix} have been proposed instead, with the
drawback that since they always lead to improper posteriors, ad hoc corrections have to be implemented
\citep{diebolt:robert:1994,roeder:wasserman:1997,stephens:1997}.

When relying instead on dependent improper priors, it is not always the case that the posterior distribution is
improper. For instance, \cite{robert:titterington:1998} provide a location-scale representation that allows for some
improper prior. In the current paper, we consider instead the genuine Jeffreys prior for the complete set of parameters
in \eqref{eq:theMix}, derived from the Fisher information matrix for the whole model. While establishing the analytical
properness of the associated posterior is beyond the goal of the current paper, we handle large enough samples to posit
that a sufficient number of observations is allocated to each component and hence the likelihood function dominates the
prior distribution. (In the event the posterior remains improper, the associated MCMC algorithm should exhibit a
transient behaviour.)

Therefore, this is an appropriate and realistic example for evaluating Delayed Acceptance since the computation of the
prior density is clearly costly, relying on many integrals of the form:
\begin{equation} 
-\int_{\mathcal{X}} \frac{\partial^2 \log \left[\sum_{i=1}^k
w_i\,f(x|\theta_i)\right]}{\partial \theta_h \partial
\theta_j}\left[\sum_{i=1}^k w_i\,f(x|\theta_i)\right] \dd x \,.
\end{equation}
Indeed, these integrals cannot be computed analytically and thus their derivation involve numerical or Monte Carlo
integration. This setting is such that the prior ratio---as opposed to the more common case of the likelihood ratio---is
the costly part of the target evaluated in the Metropolis--Hastings acceptance ratio.
Moreover, since the Jeffreys prior involves a determinant, there is no easy way to split the computation in more parts
than  ``prior $\times$ likelihood". 
Hence, the Delayed Acceptance algorithm can be
applied by simply splitting between the prior $p^J(\psi)$ and the likelihood
$\ell(\psi|x)$ ratios, the later being computed first. Moreover, since the 
proposed prior is ``non informative", its influence on the definition of the posterior 
distribution should be small with respect to the likelihood function and, then, 
computing the likelihood ratio first should not have a substantial impact on the acceptance rate. 
However, the improper nature of the prior means using a second acceptance ratio solely based on the prior
can create trapping states in practice, even though the method remains theoretically valid.  We therefore opted for
stabilising this second step by saving a small fraction of the likelihood, corresponding to $5\%$ of
the sample, to regularise this second acceptance ratio. This choice translates into Algorithm \ref{algo:mix-algo}.

\begin{algorithm}
\caption{Metropolis--Hastings with Delayed Acceptance for Mixture Models}
\label{algo:mix-algo}
Set $\ell_2(\cdot|x) = \prod \limits_{i=1}^{ \left\lfloor p n \right\rfloor } \ell(\cdot|x_i)$ and  $\ell_1(\cdot|x) =
\prod \limits_{i=\left\lfloor p n \right\rfloor+1}^{n} \ell(\cdot|x_i)$ \quad where $p\in(0,1)$\hfill\break
\begin{enumerate}
\item Simulate $\psi^\prime\sim q(\psi^\prime|\psi)$;
\item Simulate $u_1,u_2\sim\mathcal{U}(0,1)$ and set $\lambda_1=u_1\ell_1(\psi|x)$;
\item \textbf{if} $\ell_1(\psi^\prime|x)\le \lambda_1$, repeat the current parameter value and return to 1; \\
\textbf{else} set $\lambda_2=u_2 \ell_2(\psi|x) p^J(\psi)$; 
\item \textbf{if} $\ell_2(\psi^\prime|x) p^J(\psi^\prime)\ge \lambda_2$ accept $\psi^\prime$; \\ \textbf{else} repeat the current parameter value and return to 1.
\end{enumerate}
\end{algorithm}

\subsubsection{Simulation study: }

An experiment comparing a standard Metropolis--Hastings implementation with a
Metropolis--Hastings version relying on Delayed Acceptance is summarised in Table \ref{tab:Mix}. 
Data were simulated from the following Gaussian mixture model: 
\begin{equation}
f(y|\theta)=0.10\mathcal{N}(-10,2)+0.65\mathcal{N}(0,5)+0.25\mathcal{N}(15,7). 
\label{eq:mix3model}
\end{equation}


\begin{table}
    \begin{tabular}{ | l | c | c | c | c | c | c |}
    \hline
    \textbf{Algorithm} & \textbf{ESS} (aver.) & \textbf{ESS} (sd) & \textbf{ESJD} (aver.) & \textbf{ESJD} (sd) & \textbf{time} (aver.) & \textbf{time} (sd) \\ \hline
     MH & 1575.963 & 245.96 & 0.226 & 0.44 & 513.95 & 57.81 \\ \hline
	 MH + DA & 628.767 & 87.86 & 0.215 & 0.45 & 42.22 & 22.95 \\ \hline    
    \end{tabular}

\caption{Comparison using different performance indicators in the example of mixture estimation, based
on 100 replicas of the experiments according to model \eqref{eq:mix3model} with a sample size $n=500$, 
$10^5$ MH simulations and $500$ samples for the prior estimation. (``ESS" is the effective sample size,``time" is the
computational time). The actual averaged gain ($\frac{ESS_{DA}/ESS_{MH}}{time_{DA}/time_{MH}}$) is $9.58$, higher than
the ``double average'' that the table above suggests as being around $5$. } \label{tab:Mix}
\end{table}

Both the standard Metropolis--Hastings and the Delayed Acceptance version are adapted against their respective optimal
acceptance rate, which is computed to be $2\%$, given that $\delta$ is empirically established to be $0.01$ using $500$
samples for the Monte Carlo estimation of the prior.  As a consequence the MH+DA algorithm will produce less unique
samples in the total $10^5$ iterations of the chain, as reflected in the lesser ESS in Table \ref{tab:Mix}, but this is
counterbalanced by the impressive  decrease in computing time, leading again to an overall gain in terms of $\text{ESS}
/ t$ of about $9$.

%


\section{Conclusion}\label{sec:conclusion}
We introduced in this paper Delayed Acceptance, a generic and easily implemented modification of the standard Metropolis--Hastings algorithm
that splits the acceptance rate into more than one step in order to increase the computational efficiency of the
resulting MCMC, under the sole condition that the Metropolis--Hastings ratio can be factorised this way.

The choice of splitting the target distribution into parts ultimately depends on the respective costs of computing the
said parts and of reducing theoretically the overall acceptance rate and expected square jump distance (ESJD).  Still,
this generic alternative to the standard Metropolis--Hastings approach can be considered on a customary basis, since it
both requires very little modification in programming and can be easily tested against the basic version, both
empirically and theoretically by the results of \eqref{sec:convergence}.  The Delayed Acceptance algorithm presented in
\eqref{sec:zero_zero} can significantly decrease the computational time \textit{per se} as well as the overall
acceptance rate. Nevertheless, the examples presented in Section \ref{sec:examples} suggest that the gain in terms of
computational time is not linear in the reduction of the acceptance rate, especially in the presence of optimisation
techniques like \eqref{sec:optim}.

Furthermore, our Delayed Acceptance algorithm does naturally merge with the widening range of prefetching techniques, in
order to make use of parallelisation and reduce the overall computational time even more significantly.  Most settings
of interest are open to take advantage of the proposed method, if mostly in the situation of Bayesian statistics where
the target density and/or the Metropolis--Hastings ratio always allow for a natural factorisation. The case when the
likelihood function can be factorised in an useful way represents the best gain brought by our solution, in terms of
computational time, and it may easily improve even more by exploiting parallelisation techniques. 

\subsection*{Acknowledgements}

Thanks to Christophe Andrieu for a very helpful discussion on an earlier version of the manuscript. The massive help provided by Jean-Michel Marin
and Pierre Pudlo towards an implementation on a large cluster has been 
fundamental in the completion of this work. Thanks to Samuel Livingstone for suggesting the Geometric MALA example and finally thanks to Philip Nutzman for the interesting conversation and for the suggestion of the method proposed in \eqref{sec:phil}.
 Christian P. Robert research is partly financed by Agence Nationale de la Recherche (ANR, 212, rue de Bercy 75012
Paris) on the 2012--2015 ANR-11-BS01-0010 grant ``Calibration'' and by a 2010--2015 senior chair grant of Institut Universitaire de France. Marco Banterle PhD is funded by Universit\'e Paris Dauphine.


\input BGLR14.bbl

\end{document}

%% file: BGLR14.bbl
\hyphenation{Post-Script Sprin-ger}